\documentclass[conference,a4paper]{IEEEtran}
\IEEEsettopmargin{t}{30mm}
\IEEEquantizetextheight{c}
\IEEEsettextwidth{14mm}{14mm}
\IEEEsetsidemargin{c}{0mm}

\usepackage[utf8]{inputenc}
\usepackage[T1]{fontenc}
\usepackage{url}
\usepackage{ifthen}
\usepackage{cite}
\usepackage[cmex10]{amsmath} 
\usepackage{epsfig,amsfonts,subfigure,color,amsbsy}
\usepackage[nolist]{acronym}
\usepackage{graphicx,amssymb, amsthm}
\usepackage{amsmath}
\usepackage{mathrsfs}
\usepackage{tikz}
\usepackage{multirow}
\usepackage{bigstrut}
\usepackage{psfrag}
\usepackage{stackengine}
\usetikzlibrary{shapes,arrows}
\usepackage{xcolor}
\usepackage{mathtools}
\usepackage{tabularx,array}
\usepackage{mathbbol}
\usepackage{lipsum}
\usepackage{verbatim}
\mathtoolsset{showonlyrefs}

\usepackage{algorithm}
\usepackage{algpseudocode}

\DeclareMathOperator{\supp}{supp}

\DeclareMathOperator{\weight}{wt}
\newcommand{\LW}{\weight_{\scriptscriptstyle\mathsf{L}}}
\DeclareMathOperator{\dist}{d}

\newcommand{\me}{\mathrm{e}}
\newcommand{\entH}{\mathrm{H}_e}
\newcommand{\card}[1]{\left\vert{#1}\right\vert} 	
\newcommand{\floor}[1]{\lfloor{#1}\rfloor} 	
\newcommand{\ceil}[1]{\lceil{#1}\rceil} 	
\newcommand{\set}[1]{\lbrace{#1}\rbrace} 	

\newcommand{\unilee}{\mathcal{S}_{t, m}^{(n)}}
\newcommand{\leeSphere}[2]{\mathcal{S}_{#1, #2}^{(n)}}
\newcommand{\parts}[2]{\mathcal{P}_{#2}({#1})} 

\newcommand{\typeset}[2]{\mathcal{V}_{#1}^{({#2})}} 

\IEEEoverridecommandlockouts

\begin{document}

\newcommand{\old}[1]{{\color{gray}#1}}

\theoremstyle{remark}
\theoremstyle{definition}
\newtheorem{definition}{Definition}
\newtheorem{exmp}{Example}
\newtheorem{remark}{Remark}
\newtheorem{lemma}{Lemma}
\newtheorem{theorem}{Theorem}
\newtheorem{problem}{Problem}
\newtheorem{corollary}{Corollary}
\newtheorem{claim}{Claim}

\newcommand{\jbcomment}[1]{{\color{red!55!violet}Je: #1}}
\newcommand{\JB}[1]{{\color{olive!50!cyan} #1}}
\newcommand{\hb}[1]{{\color{cyan!60!gray} #1}}
\newcommand{\GL}[1]{{\color{red} #1}}
\newcommand{\GLC}[1]{{\color{red!55}GL: #1}}
\newcommand{\vecx}{\mathbf{x}}
\newcommand{\prob}{\mathbb{P}}
\newcommand{\intring}{\mathbb{Z}}


\title{On the Properties of Error Patterns in the Constant Lee Weight Channel}

 \author{%
   \IEEEauthorblockN{Jessica Bariffi\IEEEauthorrefmark{1}\IEEEauthorrefmark{2},
     Hannes Bartz\IEEEauthorrefmark{1},
     Gianluigi Liva\IEEEauthorrefmark{1},
     and Joachim Rosenthal\IEEEauthorrefmark{2}}
   \IEEEauthorblockA{\IEEEauthorrefmark{1}%
            Institute of Communication and Navigation,
            German Aerospace Center,
            82234 Wessling, Germany\\
            Email:{\tt \{jessica.bariffi,hannes.bartz,gianluigi.liva\}@dlr.de}}
   \IEEEauthorblockA{\IEEEauthorrefmark{2}%
            Institute of Mathematics,
            University of Zurich,
            CH-8057 Zürich, Switzerland\\
            Email: {\tt rosenthal@math.uzh.ch}}
 }

 \maketitle

\thispagestyle{plain}
\pagestyle{plain}


\begin{abstract}
The problem of scalar multiplication applied to vectors is considered in the Lee metric. Unlike in other metrics, the Lee weight of a vector may be increased or decreased by the product with a nonzero, nontrivial scalar.
This problem is of particular interest for cryptographic applications, like for example Lee metric code-based cryptosystems, since an attacker may use scalar multiplication to reduce the Lee weight of the error vector and thus to reduce the complexity of the corresponding generic decoder.
The scalar multiplication problem is analyzed in the asymptotic regime.
Furthermore, the construction of a vector with constant Lee weight using integer partitions is analyzed and an efficient method for drawing vectors of constant Lee weight uniformly at random from the set of all such vectors is given.
\end{abstract}

\thispagestyle{empty}
\setcounter{page}{1}


\section{Introduction}\label{sec:not}

In the late 1950s, relating to transmitting symbols from a finite prime field $\mathbb{F}_q$, the Lee metric was introduced in \cite{Ulrich,lee1958some}. Error correcting codes endowed with the Lee metric (like BCH codes, dense error-correcting codes or codes with maximum Lee distance) were constructed and applied in various different manners  \cite{prange1959use,berlekamp1966negacyclic,golomb1968algebraic,chiang1971channels,roth1994lee,etzion2010dense, alderson2013maximum}. Recently, the Lee metric was applied to DNA storage systems \cite{gabrys2017asymmetric} and considered for cryptographic applications \cite{horlemann2019information}. New families of error correcting codes endowed with the Lee metric together with an iterative decoding algorithm were proposed \cite{santini2020low} while information set decoding (ISD) in the Lee metric has been analyzed\cite{horlemann2019information, weger2020information}.

ISD is one way to solve the well-known generic (syndrome) decoding problem, which aims at decoding an arbitrary linear code efficiently without knowing or using the structure of the code. This problem is fundamental for code-based cryptography and was shown to be NP-complete in both the Hamming metric \cite{barg1994some, berlekamp1978inherent} and the Lee metric \cite{weger2020hardness}. The desirable feature of generic (syndrome) decoding is to succeed in correcting an error vector $\mathbf{e}$ as long as its corresponding weight is small, where small refers to the Gilbert-Varshamov bound \cite{gilbert1952comparison, varshamov1957estimate}. In fact, syndrome decoding has an exponential complexity in the weight of the error for both the Hamming and the Lee weight.
From an adversarial point of view, the goal is to reduce the weight of the introduced error vector in order to make the generic (syndrome) decoding problem more feasible. In fact, while the Hamming weight of a vector with entries from a finite field is invariant under multiplication with a nonzero scalar, the Lee weight of a vector can be increased or decreased by the product with a scalar. Understanding under which conditions (and with what probability) the Lee weight on the error vector $\mathbf{e}$ is reduced represents a key preliminary step in the design of Lee metric code-based cryptosystems. We will refer to this problem as \textit{scalar multiplication problem}.

In this paper, we consider an additive channel model that adds an error vector of a fixed Lee weight to the transmitted codeword. We will refer to this channel as the \textit{constant Lee weight channel}. We present an algorithm that draws a vector of length $n$ and fixed Lee weight $t$ over the ring of integers $\intring_m$ modulo $m$ uniformly at random from the set of vectors with the same parameters. Introducing errors uniformly at random is important from a cryptographic point of view in order to hide the structure of the error pattern. We will then derive the marginal distribution of the constant Lee weight channel in the limit of large block lengths $n$. This result enables to analyze how the Lee weight of a given error vector changes when multiplied by a random nonzero scalar, in the asymptotic regime. We show that, under certain conditions, the Lee weight of such an error vector will not decrease after scalar multiplication with high probability.

The paper is organized as follows. Section \ref{sec:prel} provides the notations and preliminaries needed for the course of the paper. In Section \ref{sec:LeeChannel} we introduce the constant Lee weight channel and provide a uniform construction of an error vector of given Lee weight among all possible vectors of the same Lee weight. The scalar multiplication problem is introduced in Section \ref{sec:scalarLee}. We state the problem in a finite length setting and analyze it in the asymptotic regime. Conclusions are stated in Section \ref{sec:conc}.

\section{Notation and Preliminaries}\label{sec:prel}

We denote by $\intring_m$ the ring of integers modulo $m$, where $m$ is a positive integer. To simplify the reading, vectors will be denoted by boldface lower-case letters.

\subsection{The Lee Metric}

\begin{definition}\label{def:leeweight}
    The \textit{Lee weight of a scalar} $a \in \intring_m$ is defined as\vspace{-2mm}
    \begin{align}
        \LW(a) := \min(a, m-a).
    \end{align}
    The \textit{Lee weight of a vector} $\vecx \in \intring_m^n$ of length $n$ is defined as the sum of the Lee weights of its entries, i.e.
    \begin{align}
        \LW(\vecx) := \sum_{i=1}^n \LW(x_i).
    \end{align}
\end{definition}

Note that the Lee weight of an element $a \in \intring_m$ is upper bounded by $\lfloor m/2 \rfloor$. Hence, the Lee weight of a length-$n$ vector $\mathbf{x}$ over $\intring_m$ is at most $n\cdot \lfloor m/2 \rfloor$. To simplify the notation, we define
\[
r:=\floor{m/2}.
\]
Furthermore, we observe that if $m \in \{ 2, 3 \}$ the Lee weight is equivalent to the Hamming weight.

If we consider the elements of $\intring_m$ as points placed along a circle such that the circle is divided into $m$ arcs of equal length, then the Lee distance between two distinct values $a$ and $b$ can be interpreted as the smallest number of arcs separating the two values. Therefore, the following property holds
\begin{align}\label{property:symm_leeweight}
	\LW (a) = \LW (m- a) \; \text{ for every } a \in \lbrace 1, \dots , r \rbrace .
\end{align}

The Lee distance between two vectors is defined as follows.
\begin{definition}\label{def:leedistance}
	Let $\vecx, \mathbf{y} \in \intring_m^n$. The \textit{Lee distance} between $\mathbf{x}$ and $\mathbf{y}$ is given by the Lee weight of their difference, i.e. $$\dist_L (\mathbf{x}, \mathbf{y}) := \LW(\mathbf{x} - \mathbf{y}).$$
\end{definition}
It is well-known that the Lee distance indeed induces a metric.

\subsection{Useful Results from Information Theory}
Let $X$ be a random variable over an alphabet $\mathcal{X}$ with probability distribution $P$, where $P(x):= \prob(X=x)$ with $x \in \mathcal{X}$. The entropy $H(X)$ is defined as
\begin{align}
    H(X) := -\sum_{x \in \mathcal{X}} P(x) \log(P(x)).
\end{align}
The Kullback-Leibler divergence between two distributions $Q$ and $P$ is denoted as
\begin{align}
    D(Q\, || \, P) := \sum_{x \in \mathcal{X}} Q(x) \log\left(\frac{Q(x)}{P(x)} \right)
\end{align}

\begin{theorem}[Conditional Limit Theorem~{\cite[Theorem 11.6.2]{CoverThomasBook}}]\label{thm:CondLT}
    Let $E$ be a closed convex subset of probability distributions over a given alphabet $\mathcal{X}$ and let $Q$ be a distribution not in $E$ over the same alphabet $\mathcal{X}$. Consider $X_1, \ldots, X_n$ to be discrete random variables drawn i.i.d. $\sim Q$ and let $P^\star = \arg\min_{P\in E} D(P\, || \, Q)$. Denote by $X^n$ the random sequence $(X_1, \hdots, X_n)$ and $P_{X^n}$ its empirical distribution. Then for any $a\in\mathcal{X}$
    \begin{align}\label{prob:CondLT}
        \prob \left( X_1 = a \, | \, P_{X^n} \in E \right) \longrightarrow P^\star(a)
    \end{align}
    in probability as $n$ grows large.
\end{theorem}

\subsection{Combinatorics}
\begin{definition}\label{def:partition}
    Let $t$ and $s$ be positive integers. An \textit{integer partition} of $t$ into $s$ parts is an $s$-tuple $\lambda := (\lambda_1,\dots , \lambda_s)$ of positive integers satisfying the following two properties:
    \begin{itemize}
        \item[i.] $\lambda_1 + \hdots + \lambda_s = t$,
        \item[ii.] $\lambda_1 \geq \lambda_2 \geq \hdots \geq \lambda_s$.
    \end{itemize}
    The elements $\lambda_i$ are called \textit{parts} and we say that $s$ is the \textit{length} of the partition $\lambda$.
\end{definition}
Note that the order of the parts does not matter. This means that, for instance, the tuples $(1, 1, 2)$, $(1, 2, 1)$ and $(2, 1, 1)$ are all identical and represented only by $(2, 1, 1)$. We will denote by $\Pi_\lambda$ the set of all permutations of an integer partition $\lambda$. Let $n_i$ denote the number of occurrences of a positive integer $i$ in an integer partition $\lambda$ of $t$, where $i \in \set{1, \hdots, t}$, then $\card{ \Pi_\lambda } = \binom{t}{n_1, \hdots, n_t} = \frac{t!}{n_1!\dots n_t!}$. \\
In the following, we use $\parts{t}{}$ to denote the set of integer partitions of $t$. We write $\parts{t}{k}$ instead, if we restrict $\parts{t}{}$ to those partitions with part sizes not exceeding some fixed nonnegative integer value $k$. Note that for any $\lambda \in \parts{t}{k}$ its length $\ell_\lambda$ is bounded by $\ceil{\frac{t}{k}} \leq \ell_\lambda \leq t$.\\

We will now introduce a definition describing vectors whose Lee weight decomposition is based on a given integer partition.

\begin{definition}\label{def:type_lambda}
    For a positive integer $n$ and a given partition $\lambda \in \parts{t}{r}$ of a positive integer $t$, we say that a length-$n$ vector $\vecx$ has \textit{weight decomposition $\lambda$ over $\intring_m$} if there is a one-to-one correspondence between the Lee weight of the nonzero entries of $\vecx$ and the parts of $\lambda$.
\end{definition}

\begin{exmp}
    Let $n = 5$ and let $\lambda = (2, 1, 1)$ be an integer partition of $t = 4$ over $\intring_7$. All vectors of length $n$ over $\intring_7$ consisting of one element of Lee weight $2$ and two elements of Lee weight $1$ have weight decomposition $\lambda$.
\end{exmp}

We will denote the set of all vectors of length $n$ of the same weight decomposition $\lambda \in \parts{t}{}$ by $\typeset{t, \lambda}{n}$.

\section{The Constant Lee Weight Channel}\label{sec:LeeChannel}

Let us consider a channel
\begin{align*}
    \mathbf{y} = \vecx + \mathbf{e},
\end{align*}
where $\mathbf{y}, \vecx$ and $\mathbf{e}$ are length-$n$ vectors over $\intring_m$ and the channel introduces the error vector $\mathbf{e}$ uniformly at random from the set $\unilee$ of all vectors in $\intring_m^n$ with a fixed Lee weight $t$, i.e. \[\unilee := \set{ \mathbf{e} \in \intring_m^n \, | \, \LW(\mathbf{e}) = t }.\]

\subsection{Marginal Channel Distribution}

Since certain decoder types (e.g., iterative decoders employed for low-density parity-check codes defined over integer rings) require the knowledge of the channel's marginal conditional distribution, our goal is to describe the marginal distribution $P_e$, for a generic element $E$ of the error.
\begin{lemma}\label{lemma:marginal_dist}
    The marginal distribution of a constant Lee weight channel over $\intring_m$ is given by
    $$P_e^\star = \frac{1}{\sum_{j=0}^{m-1}\exp(-\beta \LW(j))} \exp\left(-\beta\LW(e)\right),$$
    for some constant $\beta > 0$.
\end{lemma}
\begin{proof}
    Following \cite[Ch. 12]{CoverThomasBook}, we are looking for a distribution $\mathbf{P} = (P_0, \dots , P_{m-1})$ that maximizes the entropy function 
    \begin{equation}\label{eq:entropy}
        \entH(\mathbf{P}):=-\!\!\sum_{e=0, P_e\neq 0}^{m-1}\!\! P_e \log P_e
    \end{equation}
    under the constraint that the Lee weight of the vector is $t$, or equivalently, that the normalized Lee weight of the error vector is $\delta := t/n$, i.e.
    \begin{align}\label{eq:constraint} 
       \sum_{e=0}^{m-1} \LW(e) P_e=\delta.
    \end{align}
    Let us introduce a Lagrange multiplier $\beta > 0$, which is the solution to 
    \[
    \delta = \frac{(k-1)\me^{(k+1)\beta} - k\me^{k\beta} + \me^{\beta}}{(\me^{\beta k} - 1)(\me^\beta - 1)}
    \] with $k = r + 1$.
    Then the optimization problem has the following solution
    \begin{align}\label{eq:boltzmann}
        P_e^\star=\kappa\exp\left(-\beta\LW(e)\right), 
    \end{align}
    where $\kappa$ is a normalization constant enforcing $\sum_e P_e^\star=1$.
\end{proof}
The solution \eqref{eq:boltzmann} is closely related to the problem in statistical mechanics of finding the distribution of the energy state of a given system \cite{boltzmann1868studien,gibbs1902elementary, CoverThomasBook}. Here, we may interpret the energy value of the particles as the Lee weight $\LW (e)$ of an element $e \in \intring_m$. Note that for the channel law determined by Lemma \ref{lemma:marginal_dist}, the optimum decoder will seek for the codeword at minimum Lee distance from the channel output $\mathbf{y}$.

\subsection{Error Pattern Construction}

In the following we will present an algorithm that draws a vector uniformly at random from $\unilee$ for given parameters $n, t$ and $m$. The idea is inspired by the algorithm presented in \cite{santini2020low}. We start from partitioning the desired Lee weight $t$ into integer parts of size at most $r$, since the Lee weight of any $a \in \intring_m$ is at most $r$. The main difference to the algorithm presented in \cite[Lemmas 2 and 3]{santini2020low}, and crucial to design the vector uniformly at random from $\unilee$, is that the integer partition of $t$ is not chosen uniformly at random from the set of all integer partitions $\parts{t}{r}$ of $t$. In fact, picking a partition uniformly at random from $\parts{t}{r}$ yields that some of the vectors in $\unilee$ are more probable than others. Therefore, we need to understand the number of vectors with weight decomposition $\lambda$, for a fixed partition $\lambda \in \parts{t}{r}$. The following result gives an answer to this question.\\
\begin{lemma}\label{lemma:card_typeset}
    Let $n, m$ and $t$ be positive integers with $t \leq n$ and consider the set of partitions $\parts{t}{r}$ of $t$ with part sizes not exceeding $r$. For any $\lambda \in \parts{t}{r}$ the number of vectors of length $n$ over $\intring_m$ with weight decomposition $\lambda$ is given by
    \begin{align}
        \card{\typeset{t, \lambda}{n}} = \begin{cases}
            2^{\ell_\lambda} \card{\Pi_{\lambda}}\binom{n}{\ell_\lambda} & \text{ if } m \text{ is odd}, \\
            2^{\ell_\lambda - c_{r, \lambda}} \card{\Pi_{\lambda}}\binom{n}{\ell_\lambda} & \text{ else }
        \end{cases}
    \end{align}
    where $c_{r, \lambda} = \card{ \set{ i \in \set{ 1, \dots, \ell_\lambda} \, | \, \lambda_i = r } }$.
\end{lemma}

\begin{proof}
    Recall from Definition \ref{def:type_lambda} that $\typeset{t, \lambda}{n}$ consists of all length $n$ vectors $\vecx$ whose nonzero entries are in one-to-one correspondence with the parts of $\lambda$. Let $x_{i_1}, \dots , x_{i_{\ell_\lambda}}$ denote the nonzero positions of $\vecx$ and let us first consider the case where 
    \begin{align}\label{equ:cosupport_partsizes}
        \LW (x_{i_1}) = \lambda_1,\; \dots,\; \LW (x_{i_{\ell_\lambda}}) = \lambda_{\ell_\lambda}.
    \end{align}
    Finding the number of such vectors relies on the
    ``selection with repetition'' problem \cite[Section 1.2]{jukna2011extremal}, which implies that this number is exactly $\textstyle \binom{\text{number of zeros}\; + \; \text{free spaces} \; - \; 1 }{\text{free spaces} \; - \; 1}$, i.e.
    \begin{align*}
         \binom{(n-\ell_\lambda) \; + \; (\ell_\lambda+1) \; - \; 1 }{(\ell_\lambda+1) \; - \; 1} = \binom{n}{\ell_\lambda},
    \end{align*}
    where with ``free spaces'' we mean all the possible gaps in front, between and at the end of the parts of $\lambda$.\\
    If $m$ is odd, the number $n_i$ of elements in $\intring_m$ having a nonzero Lee weight $i$ is always $2$ for every possible Lee weight $i \in \set{1, \dots, r}$. Hence, there are $2^{\ell_\lambda} \binom{n}{\ell_\lambda}$ vectors satisfying \eqref{equ:cosupport_partsizes}. On the other hand, if $m$ is even, then $n_i = 2$ for $i \in \set{1, \dots , r -1}$ and $n_r = 1$. If we define $c_{r, \lambda} = \card{ \set{ i \in \set{ 1, \dots, \ell_\lambda} \, | \, \lambda_i = r } }$ to be the number of parts of $\lambda$ equal to $r$, then the number of parts of $\lambda$ that can be flipped is $2^{\ell_\lambda - c_{r, \lambda}}$. Hence, the number of vectors satisfying \eqref{equ:cosupport_partsizes} is $2^{\ell_\lambda - c_{r, \lambda}} \binom{n}{\ell_\lambda}$.\\
    Finally, since the ordering of the nonzero elements of $\vecx$ is not necessarily the same as the order of the parts of $\lambda$, we multiply $\binom{n}{\ell_\lambda}$ by the number of permutations $\card{\Pi_\lambda}$ of $\lambda$ and obtain the desired result.
\end{proof}

Finally, the actual vector construction over $\intring_m$, described in Algorithm \ref{alg:constr}, mainly consists of picking a partition $\lambda \in \parts{t}{r}$ of the Lee weight $t$ with part sizes not exceeding $r$. The probability of $\vecx \in \unilee$ with weight decomposition $\lambda \in \parts{t}{r}$ is given by 
\[
    p_\lambda := \frac{\card{\mathcal{V}_{t, \lambda}^{(n)}}}{\sum_{\Tilde{\lambda}\in \mathcal{P}_{r(t)}} \card{\mathcal{V}_{t, \Tilde{\lambda}}^{(n)}}}.
\] 
The idea is to choose the integer partition according to the probability mass function $\mathcal{X}_{t, m}^{(n)}$ defined by the probabilities $p_\lambda$, for $\lambda \in \parts{t}{r}$. We will denote this procedure by $$\lambda \overset{\mathcal{X}_{t, m}^{(n)}}{\longleftarrow} \parts{t}{r}.$$

We then randomly flip the elements of the partition modulo $m$ and assign these values to randomly chosen positions of the error vector. Choosing an element $a$ uniformly at random from a given set $\mathcal{A}$ will be denoted by $a \overset{\$}{\longleftarrow}\mathcal{A}$. We want to emphasize at this point that for fixed parameters $n, t$ and $m$ the computation of $\mathcal{X}_{t, m}^{(n)}$ needs to be done only once at the beginning, since the distribution is only dependent on these parameters and does not change anymore.

\begin{algorithm}
\caption{Drawing a vector uniformly at random from $\unilee$}
\label{alg:constr}
    \begin{algorithmic}[1]
        \Require $n, m, t \in \mathbb{N}_{>0}$, distribution $\mathcal{X}_{t, m}^{(n)}$
        \Ensure $\mathbf{e} \overset{\$}{\gets} \unilee$
        
        \State $\lambda \overset{\mathcal{X}_{t, m}^{(n)}}{\longleftarrow} \parts{t}{r}$
        \State $F = \set{f_1,\dots,  f_{\ell_\lambda}} \overset{\$}{\gets} \set{ \pm 1}^{\ell_\lambda}$
        \State $\supp (\mathbf{e} ) \overset{\$}{\gets} \set{S \subset \set{1, \dots, n} \, : \, \card{S} = \ell_\lambda}$
        \For{$i = 1, \dots, n$}
            \If{$i \in \supp(\mathbf{e})$}
                \State $e_i \gets f_i \cdot \lambda_i$
            \Else
                \State $e_i = 0$
            \EndIf
        \EndFor
        
        \State \Return $\mathrm{random\_permutation}(\mathbf{e})$
    \end{algorithmic}
\end{algorithm}

\begin{theorem}\label{thm:uniform}
    Let $n, m$ and $t$ be positive integers. Algorithm \ref{alg:constr} draws a vector uniformly at random among $\unilee$.
\end{theorem}
\begin{proof}
    First note that $\unilee = \bigsqcup_{\lambda \in \parts{t}{r}} \typeset{t, \lambda}{n}$, where $\bigsqcup$ denotes the disjoint union of sets. Hence, we want to pick $\lambda \in \parts{t}{r}$ such that all the vectors in $\unilee$ are equally probable to be drawn. The choice of $\lambda$ is decisive for the set $\typeset{t, \lambda}{n}$. Since $\card{\typeset{t, \lambda}{n}}$ changes with $\lambda$, we pick $\lambda$ according to distribution $p_\lambda$ from $\mathcal{X}_{t, m}^{(n)}$ using Lemma \ref{lemma:card_typeset} and the result follows.
\end{proof}

\section{Scalar Multiplication Problem}\label{sec:scalarLee}
While we know that the Hamming weight of a vector over a finite field is invariant under multiplication with a nonzero scalar, the Lee weight can possibly change. In this section, we analyze the behavior of the Lee weight of a vector when multiplied by a scalar. Recalling that the Lee metric coincides with the Hamming metric over $\intring_2$ and $\intring_3$, in the following we will focus only on the case where the Lee weight is different from the Hamming weight, i.e. we focus on $\intring_m$ with $m >3$.
\begin{remark}
    Even though we will not discuss the following, we want to emphasize at this point that the Hamming weight is \emph{not} invariant under multiplication with a nonzero scalar when working over a finite integer ring that is not a field.
\end{remark}

\subsection{Problem Statement}
We now establish bounds on the probability of reducing the Lee weight of a random vector by multiplying it with a random nonzero scalar.

\begin{problem}\label{prob:scalarLee}
    Consider the ring of integers $\intring_m$, with $m>3$. Given a random vector $\vecx \in \intring_m^n$ with Lee weight $\LW(\vecx) = t$ uniformly distributed in $\unilee$. Let $a$ be chosen uniformly at random from $\intring_m\backslash\set{0}$. Find the probability that the Lee weight of $a\cdot x$ is less than the Lee weight $t$ of $x$, i.e.
    \begin{align}\label{equ:proba_Problem_Lee}
        \prob\left( \LW(a\cdot x) < t \right).
    \end{align}
\end{problem}

For simplicity, let us define the following event 
\[
F:= \lbrace \LW(a\cdot \vecx) < t \rbrace.
\]
We denote by $Q_\vecx$ the empirical distribution of the entries of $\vecx$. Recall the distribution $P^\star$ defined in \eqref{eq:boltzmann}. We will rewrite $\prob (F)$ by distinguishing between vectors $\vecx$ with $Q_\vecx$ close to $P^\star$ and all others, where by ``close'' we mean with respect to the Kullback-Leibler divergence, i.e. $Q_\vecx$ satisfies $D( Q_\vecx \, || \, P^\star) < \varepsilon$ for some $\varepsilon > 0$ small. We have that
\begin{align}
    \prob (F)	\leq &\, \prob \left( \LW(a\cdot \vecx) < t \, | \, D( Q_\vecx \, || \, P^\star) < \varepsilon \right) \notag \\
	                & +  \prob \left( D( Q_\vecx \, || \, P^\star) \geq \varepsilon \right).\label{equ:upperbound_reducing_weight}
\end{align}

Note that the probability $\prob (F)$ is dependent on three parameters: the length $n$ of the constructed vector $\vecx$, the size $m$ of the integer ring and the given Lee weight $t$ of $\vecx$. The evaluation of the bound \eqref{equ:upperbound_reducing_weight} is challenging  for $m>3$, finite $n$ and generic $t$. In the following subsection we will describe how to attack the problem for $n$ large.

\subsection{Asymptotic Analysis}
Let us focus now on the asymptotic regime, i.e. where the block length $n$ tends to infinity. Note here that we let $\LW(\vecx) = t$ grow linearly with $n$. Let us denote by $U(\intring_m)$ the uniform distribution over $\intring_m$ and let $E$ be the set of probability distributions over $\intring_m$ with an average Lee weight $\delta := t/n$, i.e.
    $$E\! := \!\left\lbrace \! p = (p_0, \dots , p_{m-1}) \Big| \, \sum_{i = 0}^{m-1} \! p_i = 1 \text{ and} \sum_{i = 0}^{m-1}\! p_i\!\LW(i) = \delta \! \right\rbrace$$
Hence, a straightforward application of Theorem \ref{thm:CondLT} yields the following corollary.

\begin{corollary}\label{thm:asymptotic}
    Let $\vecx = (x_1, \hdots, x_n) \in \intring_m^n$ a random vector drawn uniformly from $\leeSphere{\delta n}{m}$. Then, for every $\varepsilon > 0$ it holds 
    \[
    \prob \left( D( Q_\vecx \, || \, P^\star) \geq \varepsilon \right) \longrightarrow 0 \text{ as } n \longrightarrow \infty.
    \]
\end{corollary}

\begin{proof}
Let $\vecx = (x_1, \hdots, x_n) \in \intring_m^n$ be a random vector whose entries are independent and uniformly distributed in $\intring_m$. The distribution of $\vecx$ is uniform on $\intring_m^n$, and hence on $\leeSphere{\delta n}{m}$. We have that 
\[
P^\star = \arg\min_{P \in E} D(P \, || \, U(\intring_m)).
\]
Then, by Theorem \ref{thm:CondLT}, we obtain the desired result.
\end{proof}

In fact, Theorem \ref{thm:asymptotic} allows to assume that the entries of a sequence $\vecx$ drawn uniformly in $\leeSphere{\delta n}{m}$ are distributed according to $P^\star$ as $n$ grows large. Hence, in the asymptotic regime, Problem \ref{prob:scalarLee} reduces to estimating the probability $\prob \left( \LW(a\cdot \vecx) \leq \LW(\vecx) \, |\, D( Q_\vecx \, || \, P^\star) < \varepsilon \right)$. In that case, we apply Definition \ref{def:leeweight} for the Lee weight of a vector $\vecx$. Then the assumption that the entries of $\vecx$ are distributed as in \eqref{eq:boltzmann} yields, in the limit of $n$ large, the following equivalent description of the desired probability\vspace{-5mm}

\begin{align}
    \lim_{n\longrightarrow\infty} \prob (F) = 
    \prob \Big( &\sum_{i=1}^{m-1}\! \me^{-\beta \LW(i)} \LW([a\cdot i]_m)\! \\  &<\! \sum_{i=1}^{m-1}\! \me^{-\beta \LW(i)} \LW(i) \Big)\label{equ:asympt_prob}
\end{align}
By Property \eqref{property:symm_leeweight}, we can run the sum only up to $r$. Nevertheless we need to distinguish between even or odd ring order $m$. In particular, for $m$ odd we rewrite \eqref{equ:asympt_prob} as
\begin{align}
    \lim_{n\longrightarrow\infty} \prob (F) = \prob \Big( 0 < \sum_{i=1}^{r} \me^{-\beta i}(i - \LW([a\cdot i]_m)) \Big)\label{equ:asympt_prob_odd}
\end{align}
whereas for $m$ even \eqref{equ:asympt_prob} is equivalent to \vspace{-5mm}

\begin{align}
    \lim_{n\longrightarrow\infty} \prob (F) = \prob \Big( 0 < &\sum_{i=1}^{r-1} 2\me^{-\beta i}(i - \LW([a\cdot i]_m))\\ &+ \me^{-\beta r}(r - \LW([a\cdot r]_m)) \Big)\label{equ:asympt_prob_even}
\end{align}
where $[a\cdot i]_m$ denotes the reduction of $a\cdot i \mod m$.\\

Since we want $\prob(F)$ to be small (or equal to zero), we need to understand under which circumstances the sums in \eqref{equ:asympt_prob_odd} and \eqref{equ:asympt_prob_even} are non-positive. Note that both $\sum_{i=1}^{r} \me^{-\beta i}(i - \LW([a\cdot i]_m))$ and $\sum_{i=1}^{r-1} 2\me^{-\beta i}(i - \LW([a\cdot i]_m)) + \me^{-\beta r}(r - \LW([a\cdot r]_m)$ are dependent on $m$ and $\beta$, where $\beta$ depends on $\delta$. If we fix these parameters, we are able to compute the sum and hence \eqref{equ:asympt_prob}. We therefore fix $m$ and evaluate the two expressions for different values of $\delta$. Let $\delta^\star$ denote the largest normalized Lee weight such that \eqref{equ:asympt_prob_odd} or rather \eqref{equ:asympt_prob_even} are equal to zero for every $\delta < \delta^\star$.
Table \ref{tab:dec_thresholds} shows the values of the threshold $\delta^\star$ for different ring orders $m$.\vspace{-2mm}

\renewcommand{\arraystretch}{1.2}
\begin{table}[H]
\caption{Maximal normalized Lee weight $\delta^\star$ over $\intring_m$ such that $ \prob (F) = 0$ as $n \longrightarrow \infty$, for some values of $m$ compared to the maximal possible normalized Lee weight $r$.}
\label{tab:dec_thresholds}\vspace{-2mm}
\begin{center}
$\begin{array}{c||ccccccccccc}
    m & 5 & 7 & 8  & 9 & 11 & 15 & 16 & 31 & 32 & 53\\
    \hline
    \hline
    r & 2 & 3 & 4 & 4 & 5 & 7 & 8 & 15 & 16 & 26\\
    \hline
    \delta^\star  & 1.2 & 1.714  & 2  & 1.962  & 2.727 & 3.310 & 4 & 7.741 & 8 & 13.245
\end{array}$
\end{center}
\end{table}
\vspace{-4mm}
Observe from Table \ref{tab:dec_thresholds} that for $m$ an odd prime power and for $\delta^\star = (m^2-1)/4m$ (i.e. the average Lee weight when choosing an element uniformly from $\intring_m$ \cite{wyner1968upper}) the Lee weight of a vector $\vecx \in \intring_m^n$ can never be reduced when multiplied by a nonzero scalar. This fact can be established by observing that the multiplication of a random variable $X$ in $\intring_m$ by $a \in \intring_m\setminus\set{0}$ induces a permutation of the distribution. Moreover, if $X$ is distributed according to $P^\star$ with $\beta >0$, the permutation that maximizes $\mathbb{E}(\LW(aX))$ is the identity, i.e., $a=1$. On the contrary, if $\beta<0$, the identity permutation ($a=1$) minimizes $\mathbb{E}(\LW(aX))$. The result follows by observing that $\beta>0$ implies that the average Lee weight is $\delta < (m^2-1)/4m$.  
    
Note that the same result follows for any $m$ if $a \in \intring_m^\times$ is a unit modulo $m$.
Moreover, if $m$ is a power of $2$, the threshold is $$\delta^\star = m/4.$$
\vspace{-5mm}

\section{Conclusions}\label{sec:conc}

In this work we have introduced an algorithm for the construction of error patterns over $\intring_m^n$ of a fixed Lee weight. The algorithm is efficient compared to straightforward approaches, which are more involved in terms of computation and memory. The proposed algorithm is based on the idea of subdividing the tasks into subtasks, which are more or less easy to solve. The procedure is dominated by the computation of the distribution used to choose the underlying integer partition of a vector's Lee weight decomposition. For a fixed Lee weight $t$, this distribution can be pre-computed. We have shown that the presented algorithm draws a vector uniformly at random among all vectors of the same length and Lee weight. This property is important for cryptographic applications in the context of Lee metric code-based cryptography in order to avoid information leakage on the structure of the error pattern. Additionally, the results on the constant-weight Lee channel together with the random construction of sequences of fixed Lee weight were used to derive the probability of reducing the Lee weight of a vector over $\intring_m^n$ when multiplying it by a random nonzero element of $\intring_m$, for the limit case where the sequence length grows large. An open problem is to characterize this probability in the finite sequence length regime.


\bibliography{IEEEabrv,biblio}


\end{document}